\documentclass{ecai}
\usepackage{times}
\usepackage{graphicx}
\usepackage{latexsym}
\usepackage{amssymb}
\usepackage{amsmath}
\usepackage{amsthm}
\usepackage{color}
\usepackage{algorithm}
\usepackage{algorithmic}
\usepackage{mathrsfs}
\usepackage{pifont}
\usepackage{comment}
\usepackage{tikz}
\usetikzlibrary{graphs,graphs.standard}
\usetikzlibrary{arrows,shapes,positioning}
\usetikzlibrary{decorations.shapes, decorations.pathreplacing}
\usetikzlibrary{quotes}
\newcommand{\cmark}{\ding{51} }
\newcommand{\xmark}{\ding{55} }

\newtheorem{definition}{Definition}
\newtheorem{thm}{Theorem}
\newtheorem{lemma}{Lemma}

\newtheorem{example}{Example}

\begin{document}

\title{False-Name-Proof Facility Location \\ on Discrete Structures}

\author{
Taiki Todo\institute{Graduate School of Information Science and Electrical Engineering, Kyushu University, Japan \and RIKEN Center for Advanced Intelligence Project, Japan, email: todo@inf.kyushu-u.ac.jp} 
\and 
Nodoka Okada\institute{Graduate School of Information Science and Electrical Engineering, Kyushu University, Japan, email: n-okada@agent.inf.kyushu-u.ac.jp} 
\and 
Makoto Yokoo\institute{Graduate School of Information Science and Electrical Engineering, Kyushu University, Japan \and RIKEN Center for Advanced Intelligence Project, Japan, email: yokoo@inf.kyushu-u.ac.jp} 
}

\maketitle
\bibliographystyle{ecai}

\begin{abstract}
 We consider the problem of locating a single facility 
 on a vertex in a given graph 
 based on agents' preferences,
 where the domain of the preferences is either single-peaked or single-dipped,
 depending on whether they want to access the facility (a public good)
 or be far from it (a public bad). 
 Our main interest is the existence of deterministic social choice functions 
 that are Pareto efficient and {\em false-name-proof},
 i.e., resistant to fake votes.
 We show that 
 regardless of whether preferences are single-peaked or single-dipped,
 such a social choice function exists (i) for any tree graph, and
 (ii) for a cycle graph if and only if its length is less than six.
 We also show that 
 when the preferences are single-peaked,
 such a social choice function exists for any ladder (i.e., $2 \times m$ grid) graph, and
 does not exist for any larger (hyper)grid.
\end{abstract}


\section{INTRODUCTION}
\label{sec:intro}

Social choice theory analyzes
how 
collective decisions are made
based on the preferences of individual agents.
Its typical application field is voting,
where each agent reports a preference ordering
over a set of alternatives
and a social choice function selects one.
One of the most well-studied criteria for social choice functions 
is robustness to agents' manipulations.
An SCF is said to be 
{\em truthful} if
no agent can benefit by telling a lie, 
and {\em false-name-proof} if 
no agent can benefit by casting multiple fake votes. 
The Gibbard-Satterthwaite theorem implies that 
any deterministic, truthful, and Pareto efficient social choice function 
must be dictatorial. Also, 
it is known that any false-name-proof social choice function must be randomized~\cite{conitzer:WINE:2008}.

Overcoming such negative results has
been a crucial research direction,
and there have been a bunch of research directions 
that overcome those negative results.
One of the most popular approaches is to restrict agents' preferences.
For example, when 
their preferences are restricted 
to being single-peaked, 
the well-known median voter schemes are truthful, 
Pareto efficient, and anonymous~\cite{moulin:PC:1980},
and their strict subclass, called {\em target rule}, is also false-name-proof~\cite{todo:AAMAS:2011}.
The model where agents' preferences are single-peaked has also 
been called the {\em facility location problem}, 
where each agent has an ideal point on an interval,
e.g., on a street, and a social choice function locates a public good, e.g., 
a train station, to which agents want to stay close.

\begin{table*}[tb]
 \caption{Summary of our contributions. 
 \cmark indicates 
 a false-name-proof and Pareto efficient social choice function,
 and \xmark indicates that
 no such social choice function exists.
 It remains open to clarify whether such a social choice function exists
 for general hypergrid graphs
 when agents' preferences are single-dipped.}
 \label{tbl:summary}
 \centering
 \setlength{\tabcolsep}{20pt}
 \def\arraystretch{1.5}
  \begin{tabular}{r||l|l|l} 
    & Tree & Cycle $C_{k}$ & Hypergrid \\ \hline \hline
    Single-Peaked & Any \cmark \cite{nehama:AAMAS:2019} & $\mathbf{k \leq 5}$ \cmark (Thm.~\ref{thm:cycle:peaked:small}) & Ladder \cmark \cite{nehama:AAMAS:2019}
    \\
    &
    & $\mathbf{k \geq 6}$ \xmark (Thm.~\ref{thm:cycle:peaked:large}) & \textbf{Other} \xmark (Thm.~\ref{thm:grid:peaked:large} and \ref{thm:cube:peaked}) 
    \\\hline
    Single-Dipped & \textbf{Any} \cmark (Thm.~\ref{thm:tree:dipped}) 
       & $\mathbf{k \leq 5}$ \cmark (Thm.~\ref{thm:cycle:dipped:small}) & Open \\
    &  & $\mathbf{k \geq 6}$ \xmark(Thm.~\ref{thm:cycle:dipped:large}) & \\ \hline
  \end{tabular}
\end{table*}

Moulin~\cite{moulin:PC:1980},
as well as many other works on facility location problems,
considered an interval as the set of possible alternatives,
where any point in the interval can be chosen by a social choice function. 
In several practical situations, however, 
the set of possible alternatives is discrete
and has 
slightly more complex underlying network structure,
which the agents' preferences also respect.
For example, in multi-criteria voting with two criteria, 
each of which has only three options, 
the underlying network is a three-by-three grid.
When we need to choose a time-slot to organize a joint meeting,
the problem resembles choosing a point on a discrete cycle.
Dokow et al.~\cite{dokow:EC:2012} studied truthful social choice functions 
on discrete lines and cycles. 
Ono et al.~\cite{ono:PRIMA:2017} considered 
false-name-proof social choice functions on a discrete line.
However, there has been very few works 
on false-name-proof social choice functions
on more complex structures (see 
Section~\ref{sec:related}).

%
In this paper we tackle the following question:
for which graph structures does a false-name-proof and Pareto efficient social choice function exist?
When the mechanism designer can arbitrarily modify the network structure
of the set of possible alternatives,
the problem 
is simplified. 
The network structure, however, is a metaphor of
a common feature among agents' preferences, where
modifying the network structure equals 
changing the domain of agents' preferences.
This is almost impossible in practice because
agents' preferences are their own private information.
The mechanism designer, therefore, first faces
the problem of verifying whether, under a {\em given} network structure
(or equally, a given preference domain),
a desirable social choice function exists.

Locating a bad is another possible extension 
of the facility location problem,
where a social choice function is required to locate a public bad,
e.g., a nuclear plant or a disposal station,
which each agent wants to 
avoid.
Agents' preferenes are therefore assumed to be single-dipped,
which is sometimes called obnoxious.
Actually, some existing works have studied truthful facility location 
with single-dipped preferences~\cite{barbera:SCW:2012}.
Nevertheless, to the best of our knowledge,
no work has dealt with both false-name-proofness
and more complex structures than a path, such as cycles.

Table~\ref{tbl:summary} summarizes our contribution.
Regardless of 
whether the preferences are 
single-peaked or single-dipped,
there is a false-name-proof and Pareto efficient social choice function
for any tree graph 
and any cycle graph of length less than six,
and there is no such social choice function
for any larger cycle graph.
%
For hypergrid graphs,
when preferences are single-peaked, 
such a social choice function exists if and only if 
the given hypergrid graph is a ladder, i.e., 
of dimension two and at least one 
of which has at most two vertices.

\section{RELATED WORKS}
\label{sec:related}

In the literature of facility location (and social choice with single-peaked preferences), 
one of the most popular direction is to design and analyze truthful social choice functions.
Moulin~\cite{moulin:PC:1980} proposed generalized median voter schemes,
which are the only deterministic,
truthful, Pareto efficient, and anonymous social choice functions.
Procaccia and Tennenholtz~\cite{procaccia:TEAC:2013}
proposed a general framework of approximate mechanism design, 
which evaluates the worst case performance
of truthful social choice functions from the perspective of competitive ratio. 
Recently, some models 
for locating multiple heterogenous facilities have also 
been studied~\cite{serafino:ECAI:2014,fong:AAAI:2018,anastasiadis:AAMAS:2018}.
Wada et al.~\cite{wada:AAMAS:2018} considered the agents who dynamically arrive and depart. 
Some research also considered facility location on grids~\cite{sui:IJCAI:2013,escoffier:ADT:2011}
and cycles~\cite{alon:MOR:2010,alon:DM:2010,dokow:EC:2012}.
Melo et al.~\cite{melo:EJOR:2009} overviewed applications 
in practical decision making.

Over the last decade,
false-name-proofness has also been 
scrutinized
in various mechanism design problems~\cite{yokoo:GEB:2004,aziz:AAMAS:2009,todo:AAMAS:2013,zhao:ECAI:2014,tsuruta:AAMAS:2015},
as a refinement of truthfulness for 
such open and anonymous environments,
as the internet.
Bu~\cite{bu:EL:2013} clarified a connection between false-name-proofness
and population monotonicity in general social choice problems.
Todo et al.~\cite{todo:AAMAS:2011} provided a complete characterization 
of false-name-proof and Pareto efficient social choice functions for the facility location problem 
with single-peaked preferences on a continuous line.
Lesca et al.~\cite{lesca:AAMAS:2014} also 
addressed false-name-proof social choice functions that are associated with monetary compensation.
Sonoda et al.~\cite{sonoda:AAAI:2016} considered the case of locating 
two homogeneous facilities on a continuous line and a circle.
Ono et al.~\cite{ono:PRIMA:2017} studied some discrete structures,
but focused on randomized social choice functions and clarifies the relation between
false-name-proofness and population monotonicity.

One of the most similar works to this paper 
is Nehama et al.~\cite{nehama:AAMAS:2019},
which also 
clarified the network structures
under which false-name-proof and Pareto efficient social choice functions exist for single-peaked preferences. 
One clear difference from ours is that, 
in their paper they proposed a new class of graphs, 
called ZV-line, as a generalization of path graphs. 
In our paper, on the other hand, we investigate well-known existing structures, 
namely tree, hypergrid, and cycle graphs.
ZV-line graphs contain any tree and {\em ladder}
(i.e., $2 \times m$-grid for arbitrary $m \geq 2$),
but do not cover any other graphs considered in this paper,
such as larger (hyper-)grid graphs and cycle graphs of lengths not equal to four.

Locating a public bad has also been widely studied
in both economics and computer science fields,
which essentially equals to consider single-dipped preferences.
Manjunath~\cite{manjunath:IJGT:2014} characterized truthful
social choice functions on an interval.
Lahiri et al.~\cite{lahiri:MSS:2017} studied the model
for locating two public bads.
Feigenbaum and Sethuraman~\cite{feigenbaum:ITEC:2015} considered
the cases where single-peaked and single-dipped preferences coexist.
Nevertheless, all 
of these works just focused on truthful social choice functions.
To the best of our knowledge, 
this paper is the very first work that considers false-name-proof facility location 
when agents' preferences are single-dipped.

\section{PRELIMINARIES}
\label{sec:prel}

In this section, we describe the formal model of 
the facility location problem considered in this paper.
Let $\Gamma := (V,E)$ be an 
undirected, connected graph, 
defined by the set $V$ of vertices and the set $E$ of edges.
The distance function $d: V^{2} \rightarrow \mathbb{N}_{\geq 0}$ 
is such that for any $v, w \in V$, 
$d(v,w) := \# \{e \in E | e \in s(v,w)\}$,
where $s(v,w)$ is the shortest path between $v$ and $w$. 
We say that a graph $\Gamma' := (V', E')$ has another graph $\Gamma := (V, E)$
as a {\em distance-preserving induced subgraph}
if $\Gamma'$ has $\Gamma$ as an induced subgraph,
where the corresponding pair of two vertices are represented as $v \equiv v'$,
and for any pair $v,w \in V$ 
and their corresponding vertices $v', w' \in V'$,
i.e., $v \equiv v'$ and $w \equiv w'$, 
it holds that $d(v', w') = d(v,w)$.
For these $\Gamma$ and $\Gamma'$, let us also denote 
$V'_{\mid \Gamma} := \{ v' \in V' \mid \exists v \in V, v' \equiv v \}$,
and 
$U \equiv U'$ for $U \subseteq V$ and $U' \subseteq V'_{\mid \Gamma}$
if $\forall u \in U$, $\exists u' \in U'$ such that $u \equiv u'$.
%

In this paper we focus on three classes of graphs, namely
tree, cycle, and hypergrid.
A {\em tree} graph is an undirected, connected and acyclic graph.
A special case of tree graphs is called as a {\em path} graph,
in which only two vertices have a degree of one and all the others have a degree of two.
Indeed, tree graphs are a simplest generalization of path graphs,
so that most of the properties
of path graphs,
such as the uniqueness of the shortest path between two vertices, 
carries over to tree graphs.
A {\em cycle} graph is an undirected and connected graph that only consists
of a single cycle.
When a cycle graph has $k$ vertices, we refer to it as $C_{k}$,
and its vertices are labeled, in a counter-clockwise order, 
from $v_{1}$ to $v_{k}$.

A {\em hypergrid} graph is a Cartesian product of 
more than one path graphs.
When a hypergrid $\Gamma$ is a Cartesian product of $k$ path graphs,
we call it a $k$-dimensional ($k$-D, in short) grid. 
In this paper, a 2-D grid is sometimes represented
by the number of vertices on each path,
as $l \times m$-grid.
In a given $k$-D grid graph,
each vertex $v$ is represented as a $k$-tuple
$(v_{k'})_{1 \leq k' \leq k}$.
Note that the $2$-D $2 \times 2$-grid is a cycle graph $C_{4}$.

Let $\mathcal{N}$ be the set of potential agents,
and let $N \subseteq \mathcal{N}$ be a set of participating agents.
Each agent $i \in N$ has a {\em type} $\theta_{i} \in V$.
When agent $i$ has type $\theta_{i}$,
agent $i$ is {\em located on} vertex $\theta_{i}$.
Let $\theta := (\theta_{i})_{i \in N} \in V^{|N|}$
denote a profile of the agents' types, and 
let $\theta_{-i} := (\theta_{i'})_{i' \neq i}$ denote 
their profile 
without $i$'s.
Given $\theta$, 
let $I(\theta) \subseteq V$ be the set of vertices on which
at least one agent is located,
i.e., $I(\theta) := \bigcup_{i \in N} \theta_{i}$.
Furthermore, given $\theta$ and vertex $v \in I(\theta)$, 
let $\theta_{-v}$ be the profile obtained by removing
all the agents at the vertex $v$ from $\theta$.
By definition, $I(\theta_{-v}) = I(\theta) \setminus \{v\}$.

Given $\Gamma$ and $v \in V$,
let $\succsim_{v}$ be the preference of the agent 
located on vertex $v$ over the set $V$ of alternatives,
where $\succ_{v}$ and $\sim_{v}$ indicate 
the strict and indifferent parts of $\succsim_{v}$,
respectively.
A preference $\succsim_{v}$ is {\em single-peaked} 
(resp.\ {\em single-dipped}) 
under $\Gamma$
if,
for any $w,x \in V$, 
$w \succ_{v} x$  
if and only if $d(v, w) < d(v, x)$
(resp.\  $d(v, w) > d(v, x)$),
and $w \sim_{v} x$
if and only if $d(v, w) = d(v, x)$.
That is, 
an agent located on $v$ strictly prefers 
alternative $w$, which is strictly closer to (resp.\ farther from) $v$
than other alternative $x$,
and is indifferent between these alternatives
when they are 
the same distance from $v$.
By definition, 
for each possible type $\theta_{i}$, 
the single-peaked (resp.\ single-dipped) 
preference is unique. 


A (deterministic) social choice function is 
a mapping from the set of possible profiles to the set of vertices.
Since each agent may pretend to be multiple agents in our model,
a social choice function must be defined for different-sized 
profiles. To describe this feature, 
we define a social choice function $f = (f_{N})_{N \subseteq \mathcal{N}}$ 
as a family of functions, where each $f_{N}$ is a mapping from $V^{|N|}$ to
$V$.
When a set $N$ of agents participates,
the social choice function $f$ uses function $f_{N}$ to
determine the outcome.
The function $f_{N}$ takes profile $\theta$ of types
jointly reported by $N$ as an input, 
and returns $f_{N}(\theta)$ as an outcome.
We denote $f_{N}$ as $f$ if it is clear from the context.
We further assume that a social choice function $f$ is anonymous,
i.e., for any input $\theta$ and its permutation $\theta'$,
$f(\theta') = f(\theta)$ holds.


We are now ready to define the two desirable properties of 
social choice functions: {\em false-name-proofness}
and {\em Pareto efficiency}.

\begin{definition}[False-Name-Proofness]
 a social choice function $f$ is said to be {\em false-name-proof} 
 if 
 for any $N$, 
 any $\theta$,
 any $i \in N$, 
 any $\theta_{i} \in V$, 
 any $\theta'_{i} \in V$, 
 any $\Phi_i \subseteq \mathcal{N} \setminus N$, 
 and 
 any $\theta_{\Phi_i} \in V^{|\Phi_i|}$, 
 it holds that 
 \[
 f(\theta) \succsim_{\theta_{i}} f(\theta'_{i}, \theta_{\Phi_i}, \theta_{-i}).
 \]
\end{definition}

The set $\Phi_{i}$ indicates the set of identities
added by $i$ for the manipulation.
The property coincides with the canonical {\em truthfulness}
when $\Phi_{i} = \emptyset$,
i.e., agent $i$ only uses one identity.

\begin{definition}[Pareto Efficiency]
 An alternative $v \in V$ is said to {\em Pareto dominate} $w \in V$ 
 under $\theta$ if both
\begin{itemize}
 \item $v \succsim_{\theta_{i}} w$ for all $i \in N$, and
 \item  $v \succ_{\theta_{j}} w$ for some $j \in N$
\end{itemize}
 hold.
 a social choice function $f$ is said to be {\em Pareto efficient} 
 if for any $N$ and any $\theta$, 
 no alternative $v \in V$ Pareto dominates $f(\theta)$.
\end{definition}

Given $\theta$, let $\text{PE}(\theta) \subseteq V$ indicate 
the set of all alternatives that are not Pareto dominated 
by any alternative.

The following theorem on a general property of false-name-proof and Pareto efficient social choice functions 
has recently been provided by the authors' another paper~\cite{okada:PRIMA:2019},
which justifies to focus on a special class of false-name-proof and Pareto efficient social choice functions.
a social choice function $f$ is said to {\em ignore duplicate ballots} 
(or satisfies IDB in short)
if for any pair $\theta, \theta'$, $I(\theta) = I(\theta')$ 
implies $f(\theta) = f(\theta')$.
In words, any social choice function that satisfies IDB cares about
whether there exists at least one agent on each vertex,
but does not care about how many agents are located in each vertex.

\begin{thm}[Okada et al.~\cite{okada:PRIMA:2019}]
 Assume that the domain of agents' preferences are 
 either single-peaked or single-dipped.
 If there is a false-name-proof and Pareto efficient social choice function $f$ that does not satisfy IDB,
 we then can find another false-name-proof and Pareto efficient social choice function $f'$ that also satisfies IDB
 and such that for any $N$, any $i \in N$ and any $\theta$,
 \[
 f'(\theta) \sim_{i} f(\theta).
 \]
%
\end{thm}

That is, for any social choice function $f$ violating IDB, 
we can find another social choice function $f'$ that is 
indifferent with $f$, for any possible input $\theta$,
from the perspective of any participating agent.
Therefore, in what follows,
we focus on such false-name-proof and Pareto efficient social choice functions $f'$ that also satisfies IDB.

\section{SINGLE-PEAKED PREFERENCES}
\label{sec:peak}

In this section, 
we focus on single-peaked preferences,
i.e., every agent prefers to have the facility closer to her.
%
It is already known that for any tree graph, 
and thus for any path graph,
a false-name-proof and Pareto efficient social choice function exists. 

\begin{thm}[Nehama et al.~\cite{nehama:AAMAS:2019}]
 \label{thm:tree:peaked}
 Assume that agents' preferences are single-peaked.
 For any tree graph,
 there is a false-name-proof and Pareto efficient social choice function.
\end{thm}

An example of such a social choice function is the target rule~\cite{klaus:TD:2001},
originally proposed for an interval, i.e., a continuous line such as $[0,1]$. 
It is shown that the target rule is false-name-proof and Pareto efficient for any tree metric~\cite{todo:AAMAS:2011}. 
Almost the same proof works for any tree graph.

In the following two subsections, we investigate the existence 
of such social choice functions for cycle and hypergrid graphs.
The two lemmata presented below are useful
to prove the impossiblity results for
single-peaked preferences.
Lemma~\ref{lem:dist-prsv} intuitivelty shows that, 
when there is no social choice function that is truthful and Pareto efficient simultaneously
for a distance-preserving induced subgraph, then the impossiblity carries over
to the original graph.
Lemma~\ref{lem:removal} shows that,
when the current alternative is still Pareto efficient
after removing a preference, then the alternative must be still chosen.
Notice that both lammata does not assume any specific structure of the graphs,
i.e., does hold not only for grids and cycles but also for general structures.

\begin{lemma}
 \label{lem:dist-prsv}
 Let $\Gamma = (V,E)$ be an arbitrary graph.
 Assume that agents' preferences are single-peaked under $\Gamma$
 and there is no truthful and Pareto efficient social choice function for $\Gamma$.
 Then, for any graph $\Gamma' = (V', E')$ that contains $\Gamma$ as a distance-preserving 
 induced subgraph, there is no truthful and Pareto efficient social choice function.
\end{lemma}

\begin{proof}
 Consider an arbitrarily chosen social choice function $f'$ for $\Gamma'$.
 Because (i) $\Gamma'$ has $\Gamma$ as a distance-preserving
 induced subgraph and (ii) agents' preferences are single-peaked, 
 for any profile $\theta'$ on $V'_{\mid \Gamma}$ and 
 corresponding profile $\theta$ on $V$ such that $I(\theta) = I(\theta')$, 
 it holds that
 \[
 \text{PE}(\theta') \equiv \text{PE}(\theta),
 \]
 that is, the structure of the set of Pareto efficient
 alternatives are totally the same.
 Therefore, for the arbitrarily chosen social choice function $f'$ for $\Gamma'$,
 its behavior for the set of profiles $\theta'$ such that $I(\theta') \subseteq V'_{\mid \Gamma}$
 must be equal to a social choice function $f$ for $\Gamma$,
 i.e., 
 there exists a social choice function $f$ for $\Gamma$ such that
 \[
 \forall \theta' \text{ s.t.\ } I(\theta) \subseteq V'_{\mid \Gamma},
 \forall \theta \text{ s.t.\ } I(\theta) = I(\theta'),
 f(\theta) \equiv f'(\theta').
 \]
 By the assumption, such a social choice function $f$ for $\Gamma$
 is not truthful and Pareto efficient simultaneously.
 Therefore, $f'$ also violate one of the properties.
\end{proof}

\begin{lemma}
 \label{lem:removal}
 Let $\Gamma$ be an arbitrary graph.
 Assume that agents' preferences are single-peaked under a graph $\Gamma$.
 Then, for any false-name-proof social choice function $f$, any $\theta$ and any $v \in I(\theta)$,
 \[
 [f(\theta) \in I(\theta) \land f(\theta) \in I(\theta_{-v})] 
 \Rightarrow
 [f(\theta_{-v}) = f(\theta)].
 \] 
\end{lemma}

\begin{proof}
 Assume for the sake of contradiction that
 the alternative chosen by the social choice function $f$
 changes after the removal of such $v$, i.e.,
 \[
 f(\theta_{-v}) \neq f(\theta).
 \]
 Here, 
 \[
 f(\theta) \in I(\theta)
 \]
 implies that there is some agent, say $i$, located at $f(\theta)$,
 who incurs the cost of zero when $\theta$ is reported
 and is still present when $v$ is removed.
 Since 
 \[
 f(\theta) \neq f(\theta_{-v})
 \]
 and agents' preferences are single-peaked,
 such an agent $i$ incurs the cost of more than zero
 when $v$ is removed.
 Thus, the agent $i$ located at $f(\theta)$ has an
 incentive to add identities at $v$,
 so that the situation becomes identical to the case
 of $\theta$, which contradicts the assumtion that
 $f$ is false-name-proofness.
\end{proof}

\subsection{Single-Peaked Preferences on Cycles}
\label{ssec:peak:cycle}

%
In this section, we show that,
under single-peaked preferences,
there is a false-name-proof and Pareto efficient social choice function
for $C_{k}$ if and only if $k \leq 5$. 

For the if direction,
to explain the existence of such social choice functions, 
we first define a class of social choice functions,
called {\em sequential Pareto} rules.
Given cycle $C_{k}$, a sequential Pareto rule has an ordering $\sigma$
of all the alternatives in $C_{k}$. For a given input $\theta$,
it sequentially checks, in the order specified by $\sigma$,
whether the first (second, third, and so on) alternative is Pareto efficient,
and terminates when it finds a Pareto efficient one.
By definition, any sequential Pareto rule is automatically Pareto efficient.

For a continuous circle,
any truthful and Pareto efficient social choice function is dictatorial~\cite{schummer:JET:2002}.
Since choosing such a dictator
in a non-manipulable manner, when there is uncertainty on identities, 
is quite difficult,
false-name-proof and Pareto efficient social choice functions are not likely to
exist for a continuous circle.
Our results in this section thus 
demonstrate the power of the discretization 
of the alternative space;
by discretizing the set of alternatives 
so that at most five alternatives exist along with a cycle,
we can avoid falling into the impossibility.

Dokow et al.~\cite{dokow:EC:2012} showed that 
any truthful and onto social choice function is nearly dictatorial for a cycle $C_{k}$ with $k \geq 22$.
In this paper we clarify a 
stricter threshold on such an impossibility 
when agents can pretend to be multiple agents;
false-name-proof social choice functions exist for a cycle $C_{k}$ if and only if $k \leq 5$.
Theorem~\ref{thm:cycle:peaked:small} shows the if direction,
and Theorem~\ref{thm:cycle:peaked:large} shows the only if direction.

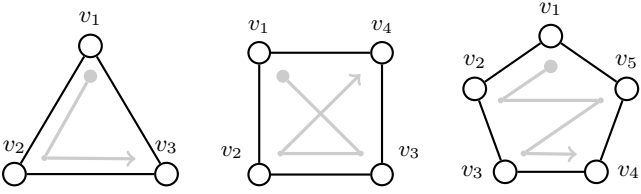
\begin{figure}[t]
 \centering
 \begin{minipage}[t]{.3\linewidth}
  \centering
 \begin{tikzpicture}[scale=1,transform shape]
  \tikzstyle{m}=[
  draw, 
  circle, 
  thick, 
  font = \scriptsize
  ];
  \tikzstyle{dummy}=[
  circle, 
  inner sep=0pt
  ];
  \node at (0,1) [style=m, "$v_{1}$"'] (1) {};
  \node at (-1,-.7) [style=m, "$v_{2}$"'] (2) {};
  \node at (1,-.7) [style=m, "$v_{3}$"'] (3) {};
  \path[draw, thick]		
  (1) edge[-] (2)
  (2) edge[-] (3)
  (3) edge[-] (1);
  \node[style=dummy, fill = gray!40, draw = gray!40, thick, minimum size = 1.5mm
  ] (1') [below = 1.5mm of 1] {}; 
  \node[style=dummy, fill =  gray!40, minimum size = .5mm, draw = gray!40, thick] (2') [above right = .7mm and 2.5mm of 2] {};
  \node[style=dummy] (3') [above left = .7mm and 2.5mm of 3] {}; 
  \path[draw, very thick, gray!40]		
  (1') edge[-] (2')
  (2') edge[->] (3');
 \end{tikzpicture}
 \end{minipage}
 \ \ \ 
 \begin{minipage}[t]{.3\linewidth} 
  \centering
 \begin{tikzpicture}[scale=.95,transform shape]
  \tikzstyle{m}=[
  draw, 
  circle, 
  thick, 
  font = \scriptsize
  ];
  \tikzstyle{dummy}=[
  circle, 
  inner sep=0pt
  ];
  \node at (-.85,.85) [style=m, label={90:$v_{1}$}] (1) {};
  \node at (-.85,-.85) [style=m, label={135:$v_{2}$}] (2) {};
  \node at (.85,-.85) [style=m, label={45:$v_{3}$}] (3) {};
  \node at (.85,.85) [style=m, label={90:$v_{4}$}] (4) {};
  \path[draw, thick]		
  (1) edge[-] (2)
  (2) edge[-] (3)
  (3) edge[-] (4)
  (4) edge[-] (1);
  \node[style=dummy, fill = gray!40, draw = gray!40, thick, minimum size = 1.5mm
  ] (1') [below right = 1.5mm and 1.5mm of 1] {}; 
  \node[style=dummy, fill =  gray!40, minimum size = .5mm, draw = gray!40, thick] (2') [above right = 1.5mm and 1.5mm of 2] {};
  \node[style=dummy, fill =  gray!40, minimum size = .5mm, draw = gray!40, thick] (3') [above left = 1.5mm and 1.5mm of 3] {};
  \node[style=dummy] (4') [below left = 1.5mm and 1.5mm of 4] {}; 
  \path[draw, very thick, gray!40]		
  (1') edge[-] (3')
  (3') edge[-] (2')
  (2') edge[->] (4');
 \end{tikzpicture}
 \end{minipage}
 \ \ \ \ \ \ 
 \begin{minipage}[t]{.3\linewidth}  
  \centering
 \begin{tikzpicture}[scale=1,transform shape]
  \tikzstyle{m}=[
  draw, 
  circle, 
  thick, 
  font = \scriptsize
  ];
  \tikzstyle{dummy}=[
  circle, 
  inner sep=0pt
  ];
  \node at (0,.8) [style=m, label={90:$v_{1}$}] (1) {};
  \node at (-1,.1) [style=m, label={90:$v_{2}$}] (2) {};
  \node at (-.6,-1) [style=m, label={180:$v_{3}$}] (3) {};
  \node at (.6,-1) [style=m, label={0:$v_{4}$}] (4) {};
  \node at (1,.1) [style=m, label={90:$v_{5}$}] (5) {};
  \path[draw, thick]		
  (1) edge[-] (2)
  (2) edge[-] (3)
  (3) edge[-] (4)
  (4) edge[-] (5)
  (5) edge[-] (1);
  \node[style=dummy, fill = gray!40, draw = gray!40, thick, minimum size = 1.5mm
  ] (1') [below = 1.5mm of 1] {}; 
  \node[style=dummy, fill =  gray!40, minimum size = .5mm, draw = gray!40, thick] (2') [below right = .1mm and 2mm of 2] {};
  \node[style=dummy, fill =  gray!40, minimum size = .5mm, draw = gray!40, thick] (3') [above right = 1mm and 1mm of 3] {};
  \node[style=dummy] (4') [above left = 1mm and 1mm of 4] {};
  \node[style=dummy, fill =  gray!40, minimum size = .5mm, draw = gray!40, thick] (5') [below left = .1mm and 2mm of 5] {}; 
  \path[draw, very thick, gray!40]		
  (1') edge[-] (2')
  (2') edge[-] (5')
  (5') edge[-] (3')
  (3') edge[->] (4');
 \end{tikzpicture}
 \end{minipage}
 \caption{Sequential Pareto rules for $C_{3}$, $C_{4}$, and $C_{5}$. The arrows indicate the associated ordering $\sigma$,
 according to which the Pareto efficiency condition will be checked one by one. 
 }
 \label{fig:small-cycles}
\end{figure}

\begin{thm}
 \label{thm:cycle:peaked:small}
 Let $\Gamma$ be a cycle graph $C_{k}$ s.t.\ $3 \leq k \leq 5$.
 When preferences are single-peaked, 
 there is a false-name-proof and Pareto efficient social choice function.
\end{thm} 

\begin{proof}
 It is obvious that
 any sequential Pareto rule is false-name-proof for $C_{3}$.
 For $C_{4}$, the sequential Pareto rule with 
 the ordering $v_{1} \rightarrow v_{3} \rightarrow v_{2} \rightarrow v_{4}$
 is false-name-proof, which is actually the same rule mentioned by Nehama et al~\cite{nehama:AAMAS:2019}.
 Finally, for $C_{5}$,
 the sequential Pareto rule with
 the ordering $v_{1} \rightarrow v_{2} \rightarrow v_{5} \rightarrow v_{3} \rightarrow v_{4}$
 is false-name-proof, which was also informally mentioned in \cite{nehama:AAMAS:2019}.
 These rules are described in Fig.~\ref{fig:small-cycles}.
\end{proof}

One might think that a sequential Pareto rule associated with 
any possible ordering is false-name-proof.
However, the following example shows that
the ordering must be carefully chosen 
to guarantee false-name-proofness
(and truthfulness as well).
Characterizing false-name-proof and Pareto efficient social choice functions for a given cycle graph
remains open.

\begin{example}
 Consider $C_{5}$ and a sequential Pareto rule $f$
 associated with ordering 
 $v_{1} \rightarrow v_{2} \rightarrow v_{3} \rightarrow v_{4} \rightarrow v_{5}$.
 Assume that there are three agents, whose types are
 $\theta = (v_{3}, v_{4}, v_{5})$. 
 Since $\text{PE}(\theta) = \{v_{3}, v_{4}, v_{5}\}$,
 $v_{3}$ is chosen as an outcome when all the agents reports truthfully,
 where the agent located at $v_{5}$ incurs the cost of $3$.
 However, 
 she can benefit by reporting $v_{1}$ as her type,
 since $\text{PE}(v_{1}, v_{3}, v_{4}) = V$,
 and thus $f(v_{1}, v_{3}, v_{4}) = v_{1}$ reduces her cost to $1$.
\end{example}

\begin{figure}[t]
  \centering
  \begin{tikzpicture}[scale=1,transform shape]
   \tikzstyle{m}=[
   draw, 
   circle, 
   fill = white,
   thick, 
   minimum size = 3mm, 
   inner sep = 0pt,
   font = \scriptsize
   ];
   \tikzstyle{m'}=[
   draw, 
   circle,
   fill = gray!50,
   thick,
   minimum size = 3mm, 
   inner sep = 0pt,
   font = \scriptsize
   ];
   \node at (0,0) [circle, draw, minimum size = 18mm] (c) {};
   \node[style = m', label = {60:$v_{1}$}] at (c.60+30) {\tiny $f$};
   \node[style = m', label = {120:$v_{2}$}] at (c.120+30) {};
   \node[style = m', label = {180:$v_{3}$}] at (c.180+30) {};
   \node[style = m', label = {240:$v_{4}$}] at (c.240+30) {};
   \node[style = m', label = {300:$v_{5}$}] at (c.300+30) {};
   \node[style = m', label = {0:$v_{6}$}] at (c.0+30) {};
   \node at (0,1.8) [] {\scriptsize \textbf{Profile} $\theta$};
   \node at (4,0) [circle, draw, minimum size = 18mm] (c) {};
   \node[style = m, label = {60:$v_{1}$}] at (c.60+30) {};
   \node[style = m', label = {120:$v_{2}$}] at (c.120+30) {};
   \node[style = m', label = {180:$v_{3}$}] at (c.180+30) {\tiny $f$};
   \node[style = m', label = {240:$v_{4}$}] at (c.240+30) {};
   \node[style = m', label = {300:$v_{5}$}] at (c.300+30) {};
   \node[style = m, label = {0:$v_{6}$}] at (c.0+30) {};
   \node at (4,1.8) [] {\scriptsize \textbf{Profile} $\theta'$};
   \node at (0,-3.5) [circle, draw, minimum size = 18mm] (c) {};
   \node[style = m, label = {60:$v_{1}$}] at (c.60+30) {};
   \node[style = m, label = {120:$v_{2}$}] at (c.120+30) {};
   \node[style = m', label = {180:$v_{3}$}] at (c.180+30) {};
   \node[style = m', label = {240:$v_{4}$}] at (c.240+30) {};
   \node[style = m', label = {300:$v_{5}$}] at (c.300+30) {\tiny $f$};
   \node[style = m', label = {0:$v_{6}$}] at (c.0+30) {};
   \node at (0,-5.3) [] {\scriptsize \textbf{Profile} $\theta''$};
   \node at (4,-3.5) [circle, draw, minimum size = 18mm] (c) {};
   \node[style = m, label = {60:$v_{1}$}] at (c.60+30) {};
   \node[style = m, label = {120:$v_{2}$}] at (c.120+30) {};
   \node[style = m', label = {180:$v_{3}$}] at (c.180+30) {};
   \node[style = m', label = {240:$v_{4}$}] at (c.240+30) {};
   \node[style = m', label = {300:$v_{5}$}] at (c.300+30) {};
   \node[style = m, label = {0:$v_{6}$}] at (c.0+30) {};
   \node at (4,-5.3) [] {\scriptsize \textbf{Profile} $\theta^{*}$};
%
  \path[
  draw, 
  line width = 1mm,
  gray!25,
  shorten > = 14mm, shorten < = 16mm]
  (0,0) edge[->] (4,0)
  (0,0) edge[->] (0,-3.5)
  (4,0) edge[->] (4,-3.5)
  (0,-3.5) edge[->] (4,-3.5);
  \end{tikzpicture}
 \caption{Type profiles used for $C_{6}$ in the proof of Theorem~\ref{thm:cycle:peaked:large}. 
On each gray vertex there is some agent, and the vertex with label $f$ must be chosen under the profile. The proof derives a contradiction on $\theta^{*}$.
 }
 \label{fig:large-cycles-6}
\end{figure}
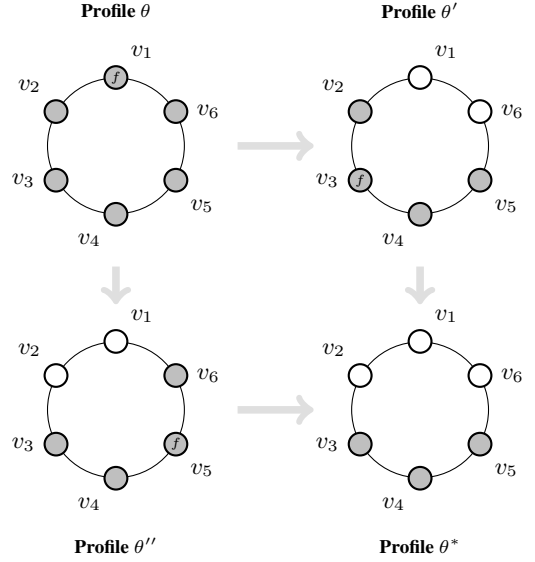

\begin{thm}
 \label{thm:cycle:peaked:large}
 Let $\Gamma$ be a cycle graph $C_{k}$ s.t.\ $k \geq 6$.
 When preferences are single-peaked, 
 there is no false-name-proof and Pareto efficient social choice function.
\end{thm}

\begin{figure}[t]
 \centering
  \begin{tikzpicture}[scale=1,transform shape]
   \tikzstyle{m}=[
   draw, 
   circle, 
   fill = white,
   thick, 
   minimum size = 3mm, 
   inner sep = 0pt,
   font = \scriptsize
   ];
   \tikzstyle{m'}=[
   draw, 
   circle,
   fill = gray!50,
   thick,
   minimum size = 3mm, 
   inner sep = 0pt,
   font = \scriptsize
   ];
   \node at (0,0) [circle, draw, minimum size = 18mm] (c) {};
   \node[style = m', label = {1*360/7+270/7:$v_{1}$}] at (c.1*360/7+270/7) {\tiny $f$};
   \node[style = m', label = {2*360/7+270/7:$v_{2}$}] at (c.2*360/7+270/7) {};
   \node[style = m', label = {3*360/7+270/7:$v_{3}$}] at (c.3*360/7+270/7) {};
   \node[style = m', label = {4*360/7+270/7:$v_{4}$}] at (c.4*360/7+270/7) {};
   \node[style = m', label = {5*360/7+270/7:$v_{5}$}] at (c.5*360/7+270/7) {};
   \node[style = m', label = {6*360/7+270/7:$v_{6}$}] at (c.6*360/7+270/7) {};
   \node[style = m', label = {7*360/7+270/7:$v_{7}$}] at (c.7*360/7+270/7) {};
   \node at (0,1.8) [] {\scriptsize \textbf{Profile} $\theta$};
   \node at (4,0) [circle, draw, minimum size = 18mm] (c) {};
   \node[style = m, label = {1*360/7+270/7:$v_{1}$}] at (c.1*360/7+270/7) {};
   \node[style = m, label = {2*360/7+270/7:$v_{2}$}] at (c.2*360/7+270/7) {};
   \node[style = m', label = {3*360/7+270/7:$v_{3}$}] at (c.3*360/7+270/7) {};
   \node[style = m', label = {4*360/7+270/7:$v_{4}$}] at (c.4*360/7+270/7) {};
   \node[style = m', label = {5*360/7+270/7:$v_{5}$}] at (c.5*360/7+270/7) {\tiny $f$};
   \node[style = m', label = {6*360/7+270/7:$v_{6}$}] at (c.6*360/7+270/7) {};
   \node[style = m, label = {7*360/7+270/7:$v_{7}$}] at (c.7*360/7+270/7) {};
   \node at (4,1.8) [] {\scriptsize \textbf{Profile} $\theta^{(1)}$};
   \node at (0,-3.5) [circle, draw, minimum size = 18mm] (c) {};
   \node[style = m, label = {1*360/7+270/7:$v_{1}$}] at (c.1*360/7+270/7) {};
   \node[style = m', label = {2*360/7+270/7:$v_{2}$}] at (c.2*360/7+270/7) {\tiny $f$};
   \node[style = m', label = {3*360/7+270/7:$v_{3}$}] at (c.3*360/7+270/7) {};
   \node[style = m', label = {4*360/7+270/7:$v_{4}$}] at (c.4*360/7+270/7) {};
   \node[style = m', label = {5*360/7+270/7:$v_{5}$}] at (c.5*360/7+270/7) {};
   \node[style = m, label = {6*360/7+270/7:$v_{6}$}] at (c.6*360/7+270/7) {};
   \node[style = m, label = {7*360/7+270/7:$v_{7}$}] at (c.7*360/7+270/7) {};
   \node at (0,-5.3) [] {\scriptsize \textbf{Profile} $\theta^{(2)}$};
   \node at (4,-3.5) [circle, draw, minimum size = 18mm] (c) {};
   \node[style = m, label = {1*360/7+270/7:$v_{1}$}] at (c.1*360/7+270/7) {};
   \node[style = m, label = {2*360/7+270/7:$v_{2}$}] at (c.2*360/7+270/7) {};
   \node[style = m', label = {3*360/7+270/7:$v_{3}$}] at (c.3*360/7+270/7) {};
   \node[style = m', label = {4*360/7+270/7:$v_{4}$}] at (c.4*360/7+270/7) {};
   \node[style = m', label = {5*360/7+270/7:$v_{5}$}] at (c.5*360/7+270/7) {};
   \node[style = m, label = {6*360/7+270/7:$v_{6}$}] at (c.6*360/7+270/7) {};
   \node[style = m, label = {7*360/7+270/7:$v_{7}$}] at (c.7*360/7+270/7) {};
   \node at (4,-5.3) [] {\scriptsize \textbf{Profile} $\theta^{*}$};
%
  \path[
  draw, 
  line width = 1mm,
  gray!25,
  shorten > = 16mm, shorten < = 13mm]
  (0,0) edge[->] (4,0)
  (0,0) edge[->] (0,-3.5)
  (4,0) edge[->] (4,-3.5)
  (0,-3.5) edge[->] (4,-3.5);
  \end{tikzpicture}
 \caption{Type profiles used for $C_{7}$ in the proof of Theorem~\ref{thm:cycle:peaked:large}.
 On each gray vertex there is some agent, and the vertex with label $f$ must be chosen under the profile. The proof derives a contradiction on $\theta^{*}$.
 }
 \label{fig:large-cycles-7}
\end{figure}
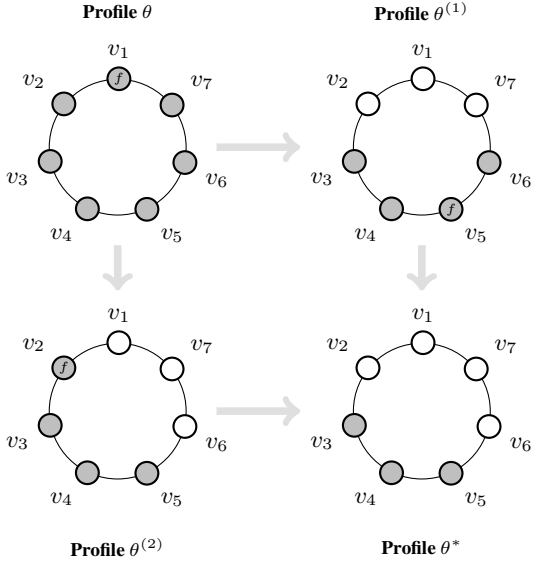

The impossibility for $C_{6}$ is important
because we will use it in the proof of Theorem~\ref{thm:cube:peaked}
in the next subsection.
Due to space limitations,
we show the proof for $C_{6}$ and $C_{7}$.

\begin{proof} 
 Assume for the sake of contradiction that
 a false-name-proof and Pareto efficient social choice function $f$ exists for 
 $C_{6}$,
 and assume w.l.o.g.\ that
 $f(\theta) = v_{1}$ for any profile $\theta$ 
 s.t.\ $I(\theta) = V$
 (see the top-left cycle of Fig.~\ref{fig:large-cycles-6}).
 
 \paragraph{For $C_{6}$:}
 Consider a type profile $\theta'$
 s.t.\ 
 \[
  I(\theta') = \{v_{2}, v_{3}, v_{4}, v_{5}\}
 \]
 (see the top-right cycle of Fig.~\ref{fig:large-cycles-6}).
 From Pareto efficienty, it must be the case that
 \[
  f(\theta') \in \text{PE}(\theta') = \{v_{2}, v_{3}, v_{4}, v_{5}\}.
 \]
 If $f(\theta') = v_{2}$, then agents at $v_{5}$ have an incentive to
 add fake identities at both $v_{1}$ and $v_{6}$.
 If $f(\theta') = v_{4}$ or $f(\theta') = v_{5}$, 
 then agents at $v_{2}$ have an incentive to
 add fake identities at both $v_{1}$ and $v_{6}$.
 Therefore, false-name-proofness implies that $f(\theta') = v_{3}$.
 Let $\theta^{*}$ then be a type profile
 s.t.\ $I(\theta^{*}) = \{v_{3}, v_{4}, v_{5}\}$,
 i.e., the antipodal to the vertex $v_{1} = f(\theta)$ 
 and its two neighbors. 
 Since $\theta^{*}$ can be obtained by removing $v_{2}$
 from $\theta'$ and $f(\theta') = v_{3} \in I(\theta^{*})$,
 Lemma~\ref{lem:removal} implies 
 \[
 f(\theta^{*}) = v_{3}.
 \]

 We then consider another profile $\theta''$ 
 s.t.\ 
 \[
  I(\theta'') = \{v_{3}, v_{4}, v_{5}, v_{6}\}
  \]
 (see the bottom-left cycle of Fig.~\ref{fig:large-cycles-6}).
 From symmetry and Lemma~\ref{lem:removal}, 
 \[
  f(\theta^{*}) = v_{5},
 \]
 which contradicts 
 $f(\theta^{*}) = v_{3}$. 
 Almost the same argument holds
 for any larger even $k$.

\paragraph{For $C_{7}$:}
 Consider a type profile $\theta^{(1)}$ 
 s.t.\ 
 \[
 I(\theta^{(1)}) = \{v_{3}, v_{4}, v_{5}, v_{6}\}
 \]
 (see the top-right cycle of Fig.~\ref{fig:large-cycles-7}).
 From Pareto efficienty, it must be the case that
 \[
  f(\theta^{(1)}) \in \text{PE}(\theta^{(1)}) = \{v_{3}, v_{4}, v_{5}, v_{6}\}.
 \]
 If $f(\theta^{(1)}) = v_{6}$, then agents at $v_{3}$ have an incentive to
 add fake identities at $v_{1}$, $v_{2}$ and $v_{7}$.
 If $f(\theta^{(1)}) = v_{3}$, then agents at $v_{6}$ have an incentive to
 add fake identities at $v_{1}$, $v_{2}$ and $v_{7}$.
 Therefore, false-name-proofness implies that
 \[
 f(\theta^{(1)}) \in \{v_{4}, v_{5}\}. 
 \]
 From the symmetry between $v_{4}$ and $v_{5}$
 in the profile $\theta^{(1)}$,
 assume w.l.o.g.\ that $f(\theta^{(1)}) = v_{5}$.
 From Lemma~\ref{lem:removal},
 removing $v_{6}$ from $\theta^{(1)}$ does not change the outcome.
 That is, for $\theta^{*}$ s.t.\ $I(\theta^{*}) = \{v_{3}, v_{4}, v_{5}\}$,
 it must be the case that 
 \[
 f(\theta^{*}) = v_{5}.
 \]

 Then let us consider another profile $\theta^{(2)}$ 
 s.t.\ 
 \[
  I(\theta^{(2)}) = \{v_{2}, v_{3}, v_{4}, v_{5}\}
 \]
 (see the bottom-left cycle of Fig.~\ref{fig:large-cycles-7}).
 Since $f$ is false-name-proof and Pareto efficient, 
 $f(\theta^{(2)}) \not \in \{v_{4}, v_{5}\}$;
 otherwise agents located on $v_{2}$ would add fake identities
 so that the outcome changes to $v_{1}$.
 Similarly, $f(\theta^{(2)}) \neq v_{3}$;
 otherwise, it must be the case that $f(\theta^{*}) = v_{3}$,
 which yields a contradiction.
 Thus, 
 \[
  f(\theta^{(2)}) = v_{2}.
 \]

 However, $f(\theta^{(2)}) = v_{2}$ implies 
 \[
  f(\theta^{*}) \neq v_{5};
 \]
 otherwise agents located on $v_{3}$ would add fake identities 
 so that the outcome changes to $v_{2}$.
 This also yields a contradiction.
 Almost the same argument holds
 for any larger odd $k$.
\end{proof}

\subsection{Single-Peaked Preferences on Hypergrids}
\label{ssec:peak:grid}

The facility location on a hypergrid graph
is a reasonable simplification of multi-criteria voting~\cite{sui:IJCAI:2013},
where each candidate has a pledge for each criteria, such as
taxation and diplomacy, that is embeddable on a hypergrid.
Each voter then has the most/least preferred point on the hypergrid.

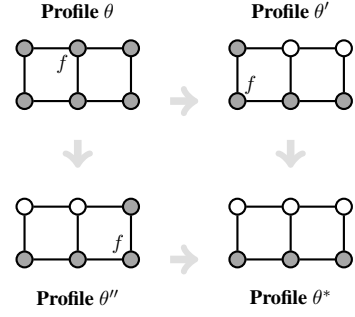
\begin{figure}[t]
 \centering
 \begin{tikzpicture}[scale=.7, transform shape, font=\large]
  \tikzstyle{m}=[
  draw, 
  circle, 
  thick, 
  minimum size = 3mm,
  inner sep = 0pt,
  font = \scriptsize
  ];
  \tikzstyle{m'}=[
  draw, 
  fill=gray!70,
  circle, 
  thick, 
  minimum size = 3mm,
  inner sep = 0pt,
  font = \scriptsize
  ];
  \node at (0,0) [style=m'] (a11) {};
  \node at (1,0) [style=m', label = {[xshift=.5mm,yshift=1mm]225:$f$}] (a12) {};
  \node at (2,0) [style=m'] (a13) {};
  \node at (0,-1) [style=m'] (a21) {};
  \node at (1,-1) [style=m'] (a22) {};
  \node at (2,-1) [style=m'] (a23) {};
  \node [above = 3mm of a12] {\textbf{Profile $\theta$}};
  \path[draw, thick]				
  (a11) edge[-] (a12)
  (a12) edge[-] (a13)
  (a21) edge[-] (a22)
  (a22) edge[-] (a23)
  (a11) edge[-] (a21)
  (a12) edge[-] (a22)
  (a13) edge[-] (a23);

  \node at (4,0) [style=m'] (b11) {};
  \node at (5,0) [style=m] (b12) {};
  \node at (6,0) [style=m] (b13) {};
  \node at (4,-1) [style=m', label = {[xshift=-1mm,yshift=-1mm]45:$f$}] (b21) {};
  \node at (5,-1) [style=m'] (b22) {};
  \node at (6,-1) [style=m'] (b23) {};
  \node [above = 3mm of b12] {\textbf{Profile $\theta'$}};
  \path[draw, thick]				
  (b11) edge[-] (b12)
  (b12) edge[-] (b13)
  (b21) edge[-] (b22)
  (b22) edge[-] (b23)
  (b11) edge[-] (b21)
  (b12) edge[-] (b22)
  (b13) edge[-] (b23);

  \node at (0,-3) [style=m] (c11) {};
  \node at (1,-3) [style=m] (c12) {};
  \node at (2,-3) [style=m'] (c13) {};
  \node at (0,-4) [style=m'] (c21) {};
  \node at (1,-4) [style=m'] (c22) {};
  \node at (2,-4) [style=m', label = {[xshift=1mm,yshift=-1mm]135:$f$}] (c23) {};
  \node [below = 3mm of c22] {\textbf{Profile $\theta''$}};
  \path[draw, thick]				
  (c11) edge[-] (c12)
  (c12) edge[-] (c13)
  (c21) edge[-] (c22)
  (c22) edge[-] (c23)
  (c11) edge[-] (c21)
  (c12) edge[-] (c22)
  (c13) edge[-] (c23);

  \node at (4,-3) [style=m] (d11) {};
  \node at (5,-3) [style=m] (d12) {};
  \node at (6,-3) [style=m] (d13) {};
  \node at (4,-4) [style=m'] (d21) {};
  \node at (5,-4) [style=m'] (d22) {};
  \node at (6,-4) [style=m'] (d23) {};
  \node [below = 3mm of d22] {\textbf{Profile $\theta^{*}$}};
  \path[draw, thick]				
  (d11) edge[-] (d12)
  (d12) edge[-] (d13)
  (d21) edge[-] (d22)
  (d22) edge[-] (d23)
  (d11) edge[-] (d21)
  (d12) edge[-] (d22)
  (d13) edge[-] (d23);

  \path[
  draw, 
  line width = 1mm,
  gray!25,
  shorten > = 4mm, shorten < = 4mm]
  (a23) edge[->] (b21)
  (b22) edge[->] (d12)
  (a22) edge[->] (c12)
  (c23) edge[->] (d21);
 \end{tikzpicture}
 \caption{Four type profiles used in the proof of Lemma~\ref{lem:3-by-2}. On each gray vertex there is some agent, and the vertex with label $f$ must be chosen under the profile. The proof derives a contradiction on $\theta^{*}$.}
 \label{fig:2-3-grid}
\end{figure}


In this section, we completely clarify
under which condition on a given hypergrid graph
a false-name-proof and Pareto efficient social choice function exists
when agents' preferences are single-peaked.

It is already known that,
when preferences are single-peaked,
a false-name-proof and Pareto efficient social choice function exists for any $2 \times m$-grid~\cite{nehama:AAMAS:2019}.
Theorem~\ref{thm:grid:peaked:large}
complements their result;
no such social choice function exists for any other 2-D grid.
Theorem~\ref{thm:cube:peaked} further shows that
this impossibility carries over into any $k$-D grid with $k \geq 3$.


\begin{lemma}
 \label{lem:3-by-2}
 Let $\Gamma$ be the $2 \times 3$-grid,
 where the set of vertices 
 $V = \{v_{1,1}, v_{1,2}, v_{1,3}, v_{2,1}, v_{2,2}, v_{2,3}\}$.
 Assume that agents' preferences are single-peaked under $\Gamma$
 and there is a false-name-proof and Pareto efficient social choice function $f$.
 Then, for any $\theta$ s.t.\  $I(\theta) = V$,
 $f(\theta)$ must be one of the four corners of $\Gamma$,
 i.e., $f(\theta) \in \{v_{1,1}, v_{1,3}, v_{2,1}, v_{2,3}\}$.
\end{lemma}

\begin{proof}
 Asssume w.l.o.g.\ that 
 $f(\theta) = v_{1,2}$
 (see the top-left grid in Fig.~\ref{fig:2-3-grid}).
 We construct a type profile $\theta'$ s.t.\ 
 $I(\theta') = \{v_{1,1}, v_{2,1}, v_{2,2}, v_{2,3}\}$.
 Since $f$ is false-name-proof,
 $f(\theta') = v_{2,1}$
 (see the top-right grid in Fig.~\ref{fig:2-3-grid}).
 We also construct another profile $\theta''$ s.t.\ 
 $I(\theta'') = \{v_{1,3}, v_{2,1}, v_{2,2}, v_{2,3}\}$.
 Since $f$ is false-name-proof,
 $f(\theta'') = v_{2,3}$
 (see the bottom-left grid in Fig.~\ref{fig:2-3-grid}).
 Finally, let $\theta^{*}$ be the profile constructed
 by removing all the agents located on $v_{1,1}$, $v_{1,2}$, 
 and $v_{1,3}$.
 By applying Lemma~\ref{lem:removal} 
 to those profiles,
 we obtain both $f(\theta^{*}) = v_{2,1}$ and $f(\theta^{*}) = v_{2,3}$,
 which yield a contradiction.
\end{proof}

Now we are ready to present the general impossibility results
for general 2-D grids (Theorem~\ref{thm:grid:peaked:large}) 
and hypergrids (Theorem~\ref{thm:cube:peaked}), 
which are our main contribution in this subsection. 

\begin{thm}
 \label{thm:grid:peaked:large}
 Let $\Gamma$ be an $l \times m$-grid, where $l, m \geq 3$.
 When preferences are single-peaked, 
 there is no false-name-proof and Pareto efficient social choice function.
\end{thm}

\begin{proof}
 Lemma~\ref{lem:3-by-3} below shows that,
 for the $3 \times 3$-grid,
 there is no false-name-proof and Pareto efficient social choice function.
 Since any $l \times m$-grid $\Gamma$, 
 for arbitrary $l, m \geq 3$,
 contains the $3 \times 3$-grid graph as a distance-preserving
 induced subgraph, the impossiblity carries over into $\Gamma$
 according to Lemma~\ref{lem:dist-prsv}.
\end{proof}

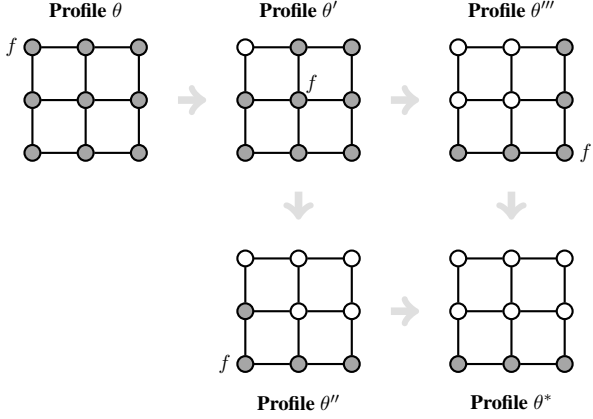
\begin{figure}[t]
 \centering
 \begin{tikzpicture}[scale=.7, transform shape, font=\large]
  \tikzstyle{m}=[
  draw, 
  circle, 
  thick, 
  font = \scriptsize
  ];
  \tikzstyle{m'}=[
  draw, 
  fill=gray!70,
  circle, 
  thick, 
  font = \scriptsize
  ];
  \node at (0,0) [style=m', label = {180:$f$}] (a11) {};
  \node at (1,0) [style=m'] (a12) {};
  \node at (2,0) [style=m'] (a13) {};
  \node at (0,-1) [style=m'] (a21) {};
  \node at (1,-1) [style=m'] (a22) {};
  \node at (2,-1) [style=m'] (a23) {};
  \node at (0,-2) [style=m'] (a31) {};
  \node at (1,-2) [style=m'] (a32) {};
  \node at (2,-2) [style=m'] (a33) {};
  \node [above = 3mm of a12] {\textbf{Profile $\theta$}};
  \path[draw, thick]				
  (a11) edge[-] (a12)
  (a12) edge[-] (a13)
  (a21) edge[-] (a22)
  (a22) edge[-] (a23)
  (a31) edge[-] (a32)
  (a32) edge[-] (a33)
  (a11) edge[-] (a21)
  (a21) edge[-] (a31)
  (a12) edge[-] (a22)
  (a22) edge[-] (a32)
  (a13) edge[-] (a23)
  (a23) edge[-] (a33);

  \node at (4,0) [style=m] (b11) {};
  \node at (5,0) [style=m'] (b12) {};
  \node at (6,0) [style=m'] (b13) {};
  \node at (4,-1) [style=m'] (b21) {};
  \node at (5,-1) [style=m', label = {[xshift=-1mm,yshift=-1mm]45:$f$}] (b22) {};
  \node at (6,-1) [style=m'] (b23) {};
  \node at (4,-2) [style=m'] (b31) {};
  \node at (5,-2) [style=m'] (b32) {};
  \node at (6,-2) [style=m'] (b33) {};
  \node [above = 3mm of b12] {\textbf{Profile $\theta'$}};
  \path[draw, thick]				
  (b11) edge[-] (b12)
  (b12) edge[-] (b13)
  (b21) edge[-] (b22)
  (b22) edge[-] (b23)
  (b31) edge[-] (b32)
  (b32) edge[-] (b33)
  (b11) edge[-] (b21)
  (b21) edge[-] (b31)
  (b12) edge[-] (b22)
  (b22) edge[-] (b32)
  (b13) edge[-] (b23)
  (b23) edge[-] (b33);

  \node at (8,0) [style=m] (c11) {};
  \node at (9,0) [style=m] (c12) {};
  \node at (10,0) [style=m'] (c13) {};
  \node at (8,-1) [style=m] (c21) {};
  \node at (9,-1) [style=m] (c22) {};
  \node at (10,-1) [style=m'] (c23) {};
  \node at (8,-2) [style=m'] (c31) {};
  \node at (9,-2) [style=m'] (c32) {};
  \node at (10,-2) [style=m', label = {0:$f$}] (c33) {};
  \node [above = 3mm of c12] {\textbf{Profile $\theta'''$}};
  \path[draw, thick]				
  (c11) edge[-] (c12)
  (c12) edge[-] (c13)
  (c21) edge[-] (c22)
  (c22) edge[-] (c23)
  (c31) edge[-] (c32)
  (c32) edge[-] (c33)
  (c11) edge[-] (c21)
  (c21) edge[-] (c31)
  (c12) edge[-] (c22)
  (c22) edge[-] (c32)
  (c13) edge[-] (c23)
  (c23) edge[-] (c33);

  \node at (4,-4) [style=m] (d11) {};
  \node at (5,-4) [style=m] (d12) {};
  \node at (6,-4) [style=m] (d13) {};
  \node at (4,-5) [style=m'] (d21) {};
  \node at (5,-5) [style=m] (d22) {};
  \node at (6,-5) [style=m] (d23) {};
  \node at (4,-6) [style=m', label = {180:$f$}] (d31) {};
  \node at (5,-6) [style=m'] (d32) {};
  \node at (6,-6) [style=m'] (d33) {};
  \node [below = 3mm of d32] {\textbf{Profile $\theta''$}};
  \path[draw, thick]				
  (d11) edge[-] (d12)
  (d12) edge[-] (d13)
  (d21) edge[-] (d22)
  (d22) edge[-] (d23)
  (d31) edge[-] (d32)
  (d32) edge[-] (d33)
  (d11) edge[-] (d21)
  (d21) edge[-] (d31)
  (d12) edge[-] (d22)
  (d22) edge[-] (d32)
  (d13) edge[-] (d23)
  (d23) edge[-] (d33);

  \node at (8,-4) [style=m] (e11) {};
  \node at (9,-4) [style=m] (e12) {};
  \node at (10,-4) [style=m] (e13) {};
  \node at (8,-5) [style=m] (e21) {};
  \node at (9,-5) [style=m] (e22) {};
  \node at (10,-5) [style=m] (e23) {};
  \node at (8,-6) [style=m'] (e31) {};
  \node at (9,-6) [style=m'] (e32) {};
  \node at (10,-6) [style=m'] (e33) {};
  \node [below = 3mm of e32] {\textbf{Profile $\theta^{*}$}};
  \path[draw, thick]				
  (e11) edge[-] (e12)
  (e12) edge[-] (e13)
  (e21) edge[-] (e22)
  (e22) edge[-] (e23)
  (e31) edge[-] (e32)
  (e32) edge[-] (e33)
  (e11) edge[-] (e21)
  (e21) edge[-] (e31)
  (e12) edge[-] (e22)
  (e22) edge[-] (e32)
  (e13) edge[-] (e23)
  (e23) edge[-] (e33);

  \path[
  draw, 
  line width = 1mm,
  gray!25,
  shorten > = 4mm, shorten < = 4mm]
  (a23) edge[->] (b21)
  (b23) edge[->] (c21)
  (d23) edge[->] (e21)
  (b32) edge[->] (d12)
  (c32) edge[->] (e12);
 \end{tikzpicture}
 \caption{Five type profiles used in the proof of Lemma~\ref{lem:3-by-3}. On each gray vertex there is some agent, and the vertex with label $f$ must be chosen under the profile. The proof derives a contradiction on $\theta^{*}$.}
 \label{fig:grid}
\end{figure}

\begin{lemma}
 \label{lem:3-by-3}
 Let $\Gamma$ be the 2-D $3 \times 3$-grid.
 When preferences are single-peaked, 
 there is no false-name-proof and Pareto efficient social choice function.
\end{lemma}

\begin{proof}
 Assume that 
 a false-name-proof and Pareto efficient social choice function $f$ exists
 for the $3 \times 3$-grid.
 From Lemmata~\ref{lem:removal} and \ref{lem:3-by-2},
 for any $\theta$ s.t.\ $I(\theta) = V$,
 $f(\theta) \in \{(1,1), (1,3), (3,1), (3,3)\}$ holds.
 From symmetry, assume w.l.o.g.\ that 
 $f(\theta) = (1,1)$ (see the top-left grid in Fig.~\ref{fig:grid}).

 We now remove all the agents located at $(1,1)$
 from the above profile $\theta$, and refer to the profile
 as $\theta'$.
 Since $f$ is false-name-proof and Pareto efficient,
 $f(\theta') = (2,2)$. 
 Here, let $\theta''$ be the profile that
 further removes all the agents located at 
 $(1,2)$, $(1,3)$, $(2,2)$, and $(2,3)$ from $\theta'$.
 Note that $I(\theta'') = \{(2,1), (3,1), (3,2), (3,3)\}$,
 and thus $f(\theta'') = (3,1)$ holds by the same argument.
 We also consider another profile, $\theta'''$,
 which is obtained by removing all the agents at
 $(1,2)$, $(2,1)$, and $(2,2)$ from $\theta'$.
 Note that $I(\theta''') = \{(1,3), (2,3), (3,1), (3,2), (3,3)\}$,
 and $f(\theta''') = (3,3)$ by the same argument.

 Then we construct $\theta^{*}$
 by removing all the agents in the vertices 
 except for $(3,1)$, $(3,2)$, and $(3,3)$ from $\theta'$.
 Since $\theta^{*}$ is reachable from both $\theta''$
 and $\theta'''$, Lemma~\ref{lem:removal} implies
 $f(\theta^{*}) = (3,1)$ and $f(\theta^{*}) = (3,3)$,
 which yields a contradiction.
\end{proof}
  

\begin{thm}
 \label{thm:cube:peaked}
 Let $\Gamma$ be an arbitrary $k (> 2)$-D grid.
 When preferences are single-peaked, 
 there is no false-name-proof and Pareto efficient social choice function.
\end{thm} 

\begin{proof}
 We can easily observe that 
 the three-dimensional 
 $2 \times 2 \times 2$-grid, a.k.a.\ the binary cube,
 contains $C_{6}$ 
 as a distance-preserving induced subgraph, e.g., the subgraph induced from the grayed vertices in Fig.~\ref{fig:cube}.
 As we showed in Theorem~\ref{thm:cycle:peaked:large} in the previous subsection, 
 there is no false-name-proof and Pareto efficient social choice function for $C_{6}$.
 Therefore, by Lemma~\ref{lem:dist-prsv}, 
 no such social choice function exists 
 for the $2 \times 2 \times 2$-grid.
 Any other larger grid (possibly of more than three dimensions)
 contains the three-dimensional $2 \times 2 \times 2$-grid,
 and thus the impossibility 
 carries over by Lemma~\ref{lem:dist-prsv}.
\end{proof}

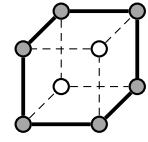
\begin{figure}[t]
 \centering
 \begin{tikzpicture}[scale=1, transform shape, font = \large]
  \tikzstyle{m}=[
  draw, 
  circle, 
  thick,
  minimum size = 2mm, inner sep=0pt,
  font = \scriptsize
  ];
  \tikzstyle{m'}=[
  draw, 
  circle, 
  fill=gray!70,
  thick, 
  minimum size = 2mm, inner sep=0pt,
  font = \scriptsize
  ];
  \node at (0, 0) [style=m'] (a1) {};
  \node at (0, -1) [style=m'] (a2) {};
  \node at (1, -1) [style=m'] (a3) {};
  \node at (1, 0) [style=m] (a4) {};
  \node at (.5, 0.5) [style=m'] (b1) {};
  \node at (.5, -.5) [style=m] (b2) {};
  \node at (1.5, -.5) [style=m'] (b3) {};
  \node at (1.5, .5) [style=m'] (b4) {};
  \path[draw, line width = .5mm, 
  ]		
  (a1) edge[-] (a2)
  (a2) edge[-] (a3)
  (a3) edge[-] (b3)
  (b3) edge[-] (b4)
  (b4) edge[-] (b1)
  (b1) edge[-] (a1);
  \path[draw, thin, densely dashed]		
  (a1) edge[-] (a2)
  (a2) edge[-] (a3)
  (a3) edge[-] (a4)
  (a4) edge[-] (a1)
  (b1) edge[-] (b2)
  (b2) edge[-] (b3)
  (b3) edge[-] (b4)
  (b4) edge[-] (b1)
  (a1) edge[-] (b1)
  (a2) edge[-] (b2)
  (a3) edge[-] (b3)
  (a4) edge[-] (b4);
 \end{tikzpicture}
 \caption{The three-dimensional $2 \times 2 \times 2$-grid, also known as a binary cube, 
 contains $C_{6}$ (emphasized by bold edges) as a distance-preserving induced subgraph, 
 which consists of the grayed vertices.}
 \label{fig:cube}
\end{figure}

\section{SINGLE-DIPPED PREFERENCES}
\label{sec:dip}

As we already mentioned in Section~\ref{sec:related}, 
this paper is the very first work that considers false-name-proof
social choice function for discrete facility location problem
when agents' preferences are single-dipped.
We therefore begin with the discussion on tree graphs.

\subsection{Single-Dipped Preferences on Trees}
\label{ssec:dip:tree}

For the case of a public bad, 
where agents' preferences are single-dipped,
we can find a false-name-proof and Pareto efficient social choice function.


\begin{thm}
 \label{thm:tree:dipped}
 Let $\Gamma$ be an arbitrary tree graph.
 When preferences are single-dipped, 
 there is a false-name-proof and Pareto efficient social choice function.
\end{thm}


\begin{proof}
 Consider the social choice function described as follows.
 First, choose an arbitrary longest path $\pi^{*}$ 
 of a given tree, whose extremes are called $a$ and $b$.
 Then, return $a$ as an outcome 
 if at least one agent strictly prefers $a$ to $b$;
 otherwise return $b$ as an outcome.

 For each agent $i$,
 either $a$ or $b$ is one of the most preferred alternative; 
 otherwise, the path from the most preferred point of $i$
 to one of the two extremes is strictly longer than $\pi^{*}$,
 which violates the assumption that $\pi^{*}$ is a longest path.
 In Fig.~\ref{fig:tree}, the agents at the bottom left gray vertex
 most prefer $b$,
 while agents at the middle or top-right gray vertices
 most prefer $a$.
 It is therefore obvious that the above social choice function is Pareto efficient,
 since either $a$ or $b$ is the most preferred alternative for each agent,
 and the choice between $a$ and $b$ is made by a unanimous voting,
 guaranteeing that the chosen alternative is the most preferred 
 for at least one agent.
 Furthermore, such a unanimous voting over two alternatives is 
 obviously false-name-proof. 
%
\end{proof}

\begin{figure}[t]
 \centering
 \begin{tikzpicture}[scale=1, transform shape, font = \large]
  \tikzstyle{m}=[
  draw, 
  circle, 
  thick, 
  font = \scriptsize
  ];
  \tikzstyle{m'}=[
  draw, 
  circle, 
  fill=gray!70,
  thick, 
  font = \scriptsize
  ];
  \node at (0, 0) [style=m, label = {90:$a$}] (1) {};
  \node at (1, 0) [style=m] (2) {};
  \node at (1.75, -0.75) [style=m] (3) {};
  \node at (1, -1.5) [style=m'] (4) {};
  \node at (2.75, -0.75) [style=m'] (5) {};
  \node at (3.5, 0) [style=m'] (6) {};
  \node at (3.5, -1.5) [style=m] (7) {};
  \node at (4.5, -1.5) [style=m, label = {90:$b$}] (8) {};
  \path[draw, line width = 1mm, 
  ]		
  (1) edge[-] (2)
  (2) edge[-] (3)
  (3) edge[-] (5)
  (5) edge[-] (7)
  (7) edge[-] (8);
  \path[draw, thin, densely dashed]		
  (1) edge[-] (2)
  (2) edge[-] (3)
  (3) edge[-] (4)
  (3) edge[-] (5)
  (5) edge[-] (6)
  (5) edge[-] (7)
  (7) edge[-] (8);
 \end{tikzpicture}
 \caption{The bold edges construct a longest path $\pi^*$. For the profile where some agent exists on each of the gray vertices, 
 the social choice function 
 described in the proof of Theorem 2 returns $a$ as an outcome.}
 \label{fig:tree}
\end{figure}
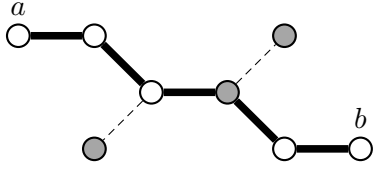

\subsection{Single-Dipped Preferences on Cycles}
\label{ssec:dip:cycle}

We next consider locating a public bad on a cycle.
Single-dipped preferences quite resemble single-peaked preferences
for cycle graphs, especially for sufficiently large ones.
Actually, in this subsection we provide 
almost the same results with the case of single-peaked preferences.

\begin{thm}
 \label{thm:cycle:dipped:small}
 Let $\Gamma$ be a cycle graph $C_{k}$ s.t.\  $3 \leq k \leq 5$.
 When preferences are single-dipped, 
 there is a false-name-proof and Pareto efficient social choice function.
\end{thm} 

\begin{proof}
 For $C_{3}$, it is easy to see that
 any seuqential Pareto rule is false-name-proof.
 For $C_{4}$, 
 the domain of single-dipped preferences
 coincides with the domain of single-peaked preferences,
 since the point diagonal from a dip point
 can be considered as a peak point.
 Therefore, the sequential Pareto rule with 
 ordering $v_{1} \rightarrow v_{3} \rightarrow v_{2} \rightarrow v_{4}$
 is false-name-proof, as shown in Theorem~\ref{thm:cycle:peaked:small}.
 Finally, for $C_{5}$, the sequential Pareto rule with 
 ordering $v_{1} \rightarrow v_{2} \rightarrow v_{5} \rightarrow v_{3} \rightarrow v_{4}$
 is false-name-proof.
\end{proof}

\begin{thm}
 \label{thm:cycle:dipped:large}
 Let $\Gamma$ be a cycle graph $C_{k}$ s.t.\ $k \geq 6$.
 When preferences are single-dipped, 
 there is no false-name-proof and Pareto efficient social choice function.
\end{thm} 


\begin{proof} 
 The 
 identical proof of Theorem~\ref{thm:cycle:peaked:large} applies
 for any even $k \geq 6$, 
 since a single-dipped preference
 over a cycle of even length, with a dip point $v$, 
 coincides with the single-peaked one 
 with the peak point that is antipodal to $v$. 
 
 We therefore focus on odd $k \geq 7$.
 Due to space limitations, we assume for the sake of contradiction 
 that
 a false-name-proof and Pareto efficient social choice function $f$ exists for 
 $C_{7}$,
 and w.l.o.g.\ that
 $f(\theta) = v_{1}$ for any $\theta$ 
 s.t.\ $I(\theta) = V$.

 Consider a type profile $\theta'$ 
 s.t.\  $I(\theta') = \{v_{1}, v_{2}, v_{6}, v_{7}\}$
 (see the top-right cycle in Fig.~\ref{fig:large-cycles-7-dipped}).
 Since $f$ is false-name-proof and Pareto efficient,
 $f(\theta')$ must be either $v_{3}$ or $v_{4}$;
 otherwise some agent has incentive to add fake identities.
 Furthermore, for the profile $\theta^{*}$ s.t.\ 
 $I(\theta^{*}) = \{v_{1}, v_{2}, v_{7}\}$,
 $\text{PE}(\theta^{*}) = \{v_{4}, v_{5}\}$ holds.
 Therefore, $f(\theta^{*}) = v_{4}$ holds;
 otherwise the agent located at $v_{7}$ has incentive to
 add fake identity on $v_{6}$, which moves the facility 
 to either $v_{3}$ or $v_{4}$.

 On the other hand, for another profile $\theta''$
 s.t.\ $I(\theta'') = \{v_{1}, v_{2}, v_{3}, v_{7}\}$
 (see the bottom-left cycle in Fig.~\ref{fig:large-cycles-7-dipped}),
 $f(\theta'')$ must be either $v_{5}$ or $v_{6}$ 
 due to symmetry. Therefore, for the above $\theta^{*}$,
 $f(\theta^{*}) = v_{5}$ must hold,
 which contradicts the condition of $f(\theta^{*}) = v_{4}$.
 Almost the same argument holds
 for any larger odd $k$.
\end{proof}  

\begin{figure}[t]
 \centering
  \begin{tikzpicture}[scale=1,transform shape]
   \tikzstyle{m}=[
   draw, 
   circle, 
   fill = white,
   thick, 
   minimum size = 3mm, 
   inner sep = 0pt,
   font = \scriptsize
   ];
   \tikzstyle{m'}=[
   draw, 
   circle,
   fill = gray!50,
   thick,
   minimum size = 3mm, 
   inner sep = 0pt,
   font = \scriptsize
   ];
   \node at (0,0) [circle, draw, minimum size = 18mm] (c) {};
   \node[style = m', label = {1*360/7+270/7:$v_{1}$}] at (c.1*360/7+270/7) {\tiny $f$};
   \node[style = m', label = {2*360/7+270/7:$v_{2}$}] at (c.2*360/7+270/7) {};
   \node[style = m', label = {3*360/7+270/7:$v_{3}$}] at (c.3*360/7+270/7) {};
   \node[style = m', label = {4*360/7+270/7:$v_{4}$}] at (c.4*360/7+270/7) {};
   \node[style = m', label = {5*360/7+270/7:$v_{5}$}] at (c.5*360/7+270/7) {};
   \node[style = m', label = {6*360/7+270/7:$v_{6}$}] at (c.6*360/7+270/7) {};
   \node[style = m', label = {7*360/7+270/7:$v_{7}$}] at (c.7*360/7+270/7) {};
   \node at (0,1.8) [] {\scriptsize \textbf{Profile} $\theta$};
   \node at (4,0) [circle, draw, minimum size = 18mm] (c) {};
   \node[style = m', label = {1*360/7+270/7:$v_{1}$}] at (c.1*360/7+270/7) {};
   \node[style = m', label = {2*360/7+270/7:$v_{2}$}] at (c.2*360/7+270/7) {};
   \node[style = m, label = {3*360/7+270/7:$v_{3}$}] at (c.3*360/7+270/7) {\tiny $f$};
   \node[style = m, label = {4*360/7+270/7:$v_{4}$}] at (c.4*360/7+270/7) {\tiny $f$};
   \node[style = m, label = {5*360/7+270/7:$v_{5}$}] at (c.5*360/7+270/7) {};
   \node[style = m', label = {6*360/7+270/7:$v_{6}$}] at (c.6*360/7+270/7) {};
   \node[style = m', label = {7*360/7+270/7:$v_{7}$}] at (c.7*360/7+270/7) {};
   \node at (4,1.8) [] {\scriptsize \textbf{Profile} $\theta'$};
   \node at (0,-3.5) [circle, draw, minimum size = 18mm] (c) {};
   \node[style = m', label = {1*360/7+270/7:$v_{1}$}] at (c.1*360/7+270/7) {};
   \node[style = m', label = {2*360/7+270/7:$v_{2}$}] at (c.2*360/7+270/7) {};
   \node[style = m', label = {3*360/7+270/7:$v_{3}$}] at (c.3*360/7+270/7) {};
   \node[style = m, label = {4*360/7+270/7:$v_{4}$}] at (c.4*360/7+270/7) {};
   \node[style = m, label = {5*360/7+270/7:$v_{5}$}] at (c.5*360/7+270/7) {\tiny $f$};
   \node[style = m, label = {6*360/7+270/7:$v_{6}$}] at (c.6*360/7+270/7) {\tiny $f$};
   \node[style = m', label = {7*360/7+270/7:$v_{7}$}] at (c.7*360/7+270/7) {};
   \node at (0,-5.3) [] {\scriptsize \textbf{Profile} $\theta''$};
   \node at (4,-3.5) [circle, draw, minimum size = 18mm] (c) {};
   \node[style = m', label = {1*360/7+270/7:$v_{1}$}] at (c.1*360/7+270/7) {};
   \node[style = m', label = {2*360/7+270/7:$v_{2}$}] at (c.2*360/7+270/7) {};
   \node[style = m, label = {3*360/7+270/7:$v_{3}$}] at (c.3*360/7+270/7) {};
   \node[style = m, label = {4*360/7+270/7:$v_{4}$}] at (c.4*360/7+270/7) {};
   \node[style = m, label = {5*360/7+270/7:$v_{5}$}] at (c.5*360/7+270/7) {};
   \node[style = m, label = {6*360/7+270/7:$v_{6}$}] at (c.6*360/7+270/7) {};
   \node[style = m', label = {7*360/7+270/7:$v_{7}$}] at (c.7*360/7+270/7) {};
   \node at (4,-5.3) [] {\scriptsize \textbf{Profile} $\theta^{*}$};
%
  \path[
  draw, 
  line width = 1mm,
  gray!25,
  shorten > = 16mm, shorten < = 13mm]
  (0,0) edge[->] (4,0)
  (0,0) edge[->] (0,-3.5)
  (4,0) edge[->] (4,-3.5)
  (0,-3.5) edge[->] (4,-3.5);
   \node at (5, -1.8) [] {\scriptsize $f(\theta^{*}) = v_{4}$};
   \node at (2, -3.2) [] {\scriptsize $f(\theta^{*}) = v_{5}$};
  \end{tikzpicture}
 \caption{Type profiles used for $C_{7}$ in the proof of Theorem~\ref{thm:cycle:dipped:large}.
 On each gray vertex there is some agent, and one of the vertices with label $f$ must be chosen under the profile. The proof derives a contradiction on $\theta^{*}$.
 }
 \label{fig:large-cycles-7-dipped}
\end{figure}
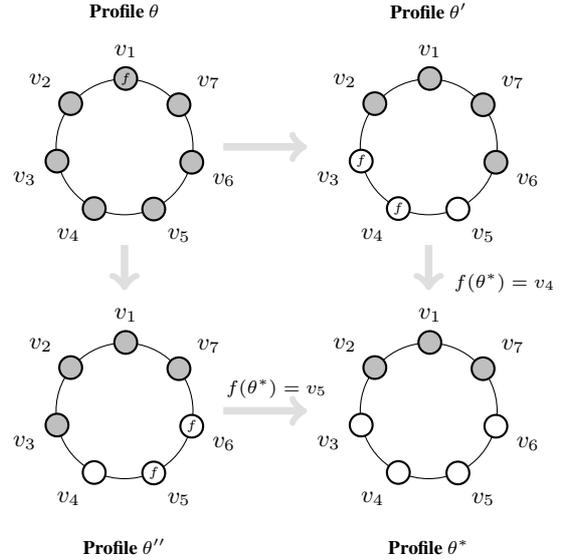

\section{CONCLUSIONS}
\label{sec:conclu}

We tackled 
whether 
there exists a false-name-proof and Pareto efficient 
social choice function for the facility location problem 
under a given graph.
We gave complete answers for path, tree, and cycle graphs,
regardless 
whether the preferences are single-peaked or
single-dipped.
For hypergrid graphs, 
an open problem remains
for single-dipped preferences.
%
When such social choice functions exist,
completely characterizing their 
class of such social choice functions is 
crucial future work,
as many other works did for continuous structure~\cite{moulin:PC:1980,schummer:JET:2002,todo:AAMAS:2011}.
Investigating randomized social choice functions is another interesting direction.

\ack 
This work is partially supported by JSPS KAKENHI Grants JP17H00761 and JP17H04695, and JST SICORP JPMJSC1607. 
%





\begin{thebibliography}{10}

\bibitem{alon:MOR:2010}
Noga Alon, Michal Feldman, Ariel~D. Procaccia, and Moshe Tennenholtz,
  `Strategyproof approximation of the minimax on networks', {\em Mathematics of
  Operations Research}, {\bf 35}(3),  513--526, (2010).

\bibitem{alon:DM:2010}
Noga Alon, Michal Feldman, Ariel~D. Procaccia, and Moshe Tennenholtz, `Walking
  in circles', {\em Discrete Mathematics}, {\bf 310}(23),  3432--3435,
  (2010).

\bibitem{anastasiadis:AAMAS:2018}
Eleftherios Anastasiadis and Argyrios Deligkas, `Heterogeneous facility
  location games', {\em Proc.\ the 17th International Conference on
  Autonomous Agents and MultiAgent Systems (AAMAS '18)}, pp. 623--631, (2018). 

\bibitem{aziz:AAMAS:2009}
Haris Aziz and Mike Paterson, `False name manipulations in weighted voting
  games: splitting, merging and annexation', {\em Proc.\ the Eighth
  International Joint Conference on Autonomous Agents and Multiagent Systems (AAMAS '09)},
  pp. 409--416, (2009).

\bibitem{barbera:SCW:2012}
Salvador Barber{\`a}, Dolors Berga, and Bernardo Moreno, `Domains, ranges and
  strategy-proofness: the case of single-dipped preferences', {\em Social
  Choice and Welfare}, {\bf 39}(2),  335--352, (2012).

\bibitem{bu:EL:2013}
Nanyang Bu, `Unfolding the mystery of false-name-proofness', {\em Economics
  Letters}, {\bf 120}(3),  559--561, (2013).

\bibitem{conitzer:WINE:2008}
Vincent Conitzer, `Anonymity-proof voting rules', {\em Proc.\ the Fourth Workshop on Internet and Network
  Economics (WINE '08)}, pp. 295--306, (2008).

\bibitem{dokow:EC:2012}
Elad Dokow, Michal Feldman, Reshef Meir, and Ilan Nehama, `Mechanism design on
  discrete lines and cycles', {\em Proc.\ the 13th {ACM} Conference on Electronic Commerce
  (EC '12)}, pp. 423--440, (2012).

\bibitem{escoffier:ADT:2011}
Bruno Escoffier, Laurent Gourv{\`e}s, Nguyen Kim~Thang, Fanny Pascual, and
  Olivier Spanjaard, `Strategy-proof mechanisms for facility location games
  with many facilities', {\em Proc.\ the Second International
  Conference on Algorithmic Decision Theory (ADT '11)}, pp. 67--81, (2011).

\bibitem{feigenbaum:ITEC:2015}
Itai Feigenbaum and Jay Sethuraman, `Strategyproof mechanisms for
  one-dimensional hybrid and obnoxious facility location models', {\em
  Proc.\ the 2015 {AAAI} Workshop on Incentive and Trust in E-Communities}, (2015).

\bibitem{fong:AAAI:2018}
Chi Kit~Ken Fong, Minming Li, Pinyan Lu, Taiki Todo, and Makoto Yokoo,
  `Facility location game with fractional preferences', {\em Proc.\ 
	the 32nd AAAI Conference on Artificial Intelligence (AAAI '18)}, 
	pp. 1039--1046, (2018).

\bibitem{klaus:TD:2001}
Bettina Klaus, `Target rules for public choice economies on tree networks and
  in euclidean spaces', {\em Theory and Decision}, {\bf 51}(1),  13--29, (2001).

\bibitem{lahiri:MSS:2017}
Abhinaba Lahiri, Hans Peters, and Ton Storcken, `Strategy-proof location of
  public bads in a two-country model', {\em Mathematical Social Sciences}, {\bf
  90},  150--159, (2017).

\bibitem{lesca:AAMAS:2014}
Julien Lesca, Taiki Todo, and Makoto Yokoo, `Coexistence of utilitarian
  efficiency and false-name-proofness in social choice', {\em Proc.\ the 13th
	International Conference on Autonomous Agents and Multi-Agent Systems (AAMAS '14)}, pp. 1201--1208, (2014).

\bibitem{manjunath:IJGT:2014}
Vikram Manjunath, `Efficient and strategy-proof social choice when preferences
  are single-dipped', {\em International Journal of Game Theory}, {\bf 43}(3),
  579--597, (2014).

\bibitem{melo:EJOR:2009}
M.~Teresa Melo, Stefan Nickel, and Francisco Saldanha{-}da{-}Gama, `Facility
  location and supply chain management -- A review', {\em European
  Journal of Operational Research}, {\bf 196}(2),  401--412, (2009).

\bibitem{moulin:PC:1980}
Herv{\'e} Moulin, `On strategy-proofness and single peakedness', {\em Public
  Choice}, {\bf 35}(4),  437--455, (1980).

\bibitem{nehama:AAMAS:2019}
Ilan Nehama, Taiki Todo, and Makoto Yokoo, `Manipulations-resistant facility
  location mechanisms for {ZV}-line graphs', {\em Proc.\ the
  18th International Conference on Autonomous Agents and Multi-Agent
  Systems (AAMAS '19)}, pp. 1452--1460, (2019).

\bibitem{okada:PRIMA:2019}
Nodoka Okada, Taiki Todo, and Makoto Yokoo, `SAT-based automated mechanism
  design for false-name-proof facility location', {\em Proc.\ the 22nd International Conference on
  Principles and Practice of Multi-Agent Systems (PRIMA '19)}, pp. 321--337,
  (2019).

\bibitem{ono:PRIMA:2017}
Tomohiro Ono, Taiki Todo, and Makoto Yokoo, `Rename and false-name
  manipulations in discrete facility location with optional preferences',
  {\em Proc.\ the 20th International Conference on Principles and Practice 
	of Multi-Agent Systems (PRIMA '17)}, pp. 163--179, (2017).

\bibitem{procaccia:TEAC:2013}
Ariel~D. Procaccia and Moshe Tennenholtz, `Approximate mechanism design without
  money', {\em {ACM} Transactions on Economics and Computation}, {\bf 1}(4),
  18:1--18:26, (2013).

\bibitem{schummer:JET:2002}
James Schummer and Rakesh~V. Vohra, `Strategy-proof location on a network',
  {\em Journal of Economic Theory}, {\bf 104}(2), 405--428, (2002).

\bibitem{serafino:ECAI:2014}
Paolo Serafino and Carmine Ventre, `Heterogeneous facility location without
  money on the line', {\em Proc.\ the 21st European Conference on
  Artificial Intelligence (ECAI '14)}, pp.
  807--812, (2014).

\bibitem{sonoda:AAAI:2016}
Akihisa Sonoda, Taiki Todo, and Makoto Yokoo, `False-name-proof locations of
  two facilities: Economic and algorithmic approachess', {\em Proc.
  the 30th {AAAI} Conference on Artificial Intelligence (AAAI '16)}, pp. 615--621, (2016).

\bibitem{sui:IJCAI:2013}
Xin Sui, Craig Boutilier, and Tuomas Sandholm, `Analysis and optimization of
  multi-dimensional percentile mechanisms', {\em Proc.\ 
  the 23rd International Joint Conference on Artificial Intelligence (IJCAI '13)}, pp. 367--374, (2013).

\bibitem{todo:AAMAS:2013}
Taiki Todo and Vincent Conitzer, `False-name-proof matching', {\em Proc.\ 
  the 12th International Conference on Autonomous Agents and Multi-Agent Systems (AAMAS '13)}, pp. 311--318, (2013).

\bibitem{todo:AAMAS:2011}
Taiki Todo, Atsushi Iwasaki, and Makoto Yokoo, `False-name-proof mechanism
  design without money', {\em Proc.\ 10th International Conference on Autonomous
  Agents and Multiagent Systems (AAMAS '11)}, pp. 651--658, (2011).

\bibitem{tsuruta:AAMAS:2015}
Shunsuke Tsuruta, Masaaki Oka, Taiki Todo, Yuko Sakurai, and Makoto Yokoo,
  `Fairness and false-name manipulations in randomized cake cutting', {\em
	Proc.\  the 14th International Conference on Autonomous Agents and
  Multiagent Systems (AAMAS '15)}, pp. 909--917, (2015).

\bibitem{wada:AAMAS:2018}
Yuho Wada, Tomohiro Ono, Taiki Todo, and Makoto Yokoo, `Facility location with
  variable and dynamic populations', {\em Proc.\ the 17th
  International Conference on Autonomous Agents and MultiAgent Systems (AAMAS '18)}, pp. 336--344, (2018).

\bibitem{yokoo:GEB:2004}
Makoto Yokoo, Yuko Sakurai, and Shigeo Matsubara, `The effect of false-name
  bids in combinatorial auctions: new fraud in internet auctions', {\em Games
  and Economic Behavior}, {\bf 46}(1),  174--188, (2004).

\bibitem{zhao:ECAI:2014}
Dengji Zhao, Siqi Luo, Taiki Todo, and Makoto Yokoo, `False-name-proof
  combinatorial auction design via single-minded decomposition', {\em Proc.\ 
	21st European Conference on Artificial Intelligence (ECAI '14)}, pp. 945--950, (2014).

\end{thebibliography}

\end{document}